\documentclass[letterpaper, 10 pt, conference]{ieeeconf}

\IEEEoverridecommandlockouts
\usepackage{cite}
\usepackage{graphicx}
\usepackage{amsmath,mathtools}
\usepackage{amsthm}
\usepackage{amssymb}
\usepackage{xcolor}
\usepackage{bm}
\usepackage[hidelinks]{hyperref}

\newtheorem{theorem}{Theorem}

\newtheorem{lemma}{Lemma}

\newtheorem{remark}{Remark}

\title{\LARGE \bf
Coupled Logarithmic Error Dynamics and a Gravity-Mismatch Bound for Thrusting Spacecraft on $\mathrm{SE}_2(3)$}

\author{Abigaile E.~Woodbury$^{1,*}$, Micah K.~Condie$^{1,*}$, Kartik A.~Pant$^{1}$, Li-Yu~Lin$^{2}$, 
\\Michael~Walker$^{3}$, and James~Goppert$^{1}$.%
\thanks{$^{1}$A.~E.~Woodbury, M.~K.~Condie, K.~A.~Pant are Graduate Research Assistants, and J.~Goppert is a Research Assistant Professor, with the School of Aeronautics and Astronautics, Purdue University. ({\tt\small condiem@purdue.edu}, {\tt\small awoodbu@purdue.edu}, {\tt\small kpant@purdue.edu}, {\tt\small jgoppert@purdue.edu})}%
\thanks{$^{2}$L.-Y.~Lin is an Assistant Professor with the Department of Aeronautics and Astronautics, National Cheng Kung University, Taiwan. ({\tt\small lylin@gs.ncku.edu.tw})}%
\thanks{$^{3}$M.~Walker is Senior Technical Director with Cromulence, Melbourne, FL, USA. ({\tt\small mike.walker@cromulence.com})}%
\thanks{$*$ denotes equal contribution by authors.}%
}

\begin{document}

\maketitle
\thispagestyle{empty}
\pagestyle{empty}

\begin{abstract}
Body-fixed spacecraft thrust couples attitude tracking error to
translational motion because the applied inertial acceleration depends
on spacecraft attitude. In this work, we formulate the dynamics of a
continuously thrusting spacecraft on $\mathrm{SE}_2(3)$ to make this
coupling explicit while retaining a structured form suitable for
analysis. We represent attitude, velocity, and position within a single
Lie-group state and derive the logarithmic dynamics of a left-invariant
tracking error with position-dependent gravity. The resulting dynamics
show explicitly how attitude error affects translational velocity
through the applied acceleration and separate the structured linear
dynamics from nonlinear contributions due to body-frame input mismatch
and position-dependent gravity difference. We specialize the gravity contribution
to the two-body model and derive an analytical bound in terms of the
reference orbital radius, logarithmic position error, and attitude
error. Numerical simulations of a continuous-thrust maneuver in a
highly elliptical Molniya-like orbit show agreement between the derived
logarithmic error dynamics and errors reconstructed from independent
nonlinear state propagation. 
\end{abstract}

\section{Introduction}

Linear relative-motion models are widely used in spacecraft guidance and control because their structure enables efficient trajectory planning, control design, and uncertainty analysis. Classical relative-motion models have been used extensively for spacecraft trajectory planning and rendezvous guidance~\cite{ClohessyWiltshire1960,Yamanaka2002,SullivanGrimbergDamico2017}. These approaches are particularly effective when attitude tracking can be treated independently from the translational problem. For spacecraft with body-fixed thrust directions, however, the inertial acceleration produced during a maneuver depends on the spacecraft attitude~\cite{MarkleyCrassidis2014}. A pointing error can therefore change the applied thrust direction and produce a translational maneuver error. This coupling becomes increasingly relevant during extended thrusting, where attitude error can affect the translational state throughout the maneuver. A model intended to analyze these errors together should therefore retain the attitude--translation coupling while preserving a structure suitable for guidance and control analysis.

Existing translational models remain effective for trajectory planning and rendezvous guidance when attitude tracking can be treated independently from the translational problem~\cite{Alfriend2010}. During body-fixed thrusting, however, this separation can hide an important source of maneuver error: the commanded acceleration is realized in the inertial frame through the spacecraft attitude. Consequently, a pointing error rotates the applied thrust away from its commanded direction, producing a velocity error that accumulates into position error over the maneuver. This effect is especially relevant for extended burns, low-thrust maneuvers, or precision proximity operations, where even a bounded attitude-tracking error acts continuously on the translational state. For such maneuvers, we seek an error model in which attitude error is not treated as an external or separate effect, but appears explicitly as a driver of translational tracking error while preserving a structure suitable for analysis.

Lie-group methods provide a geometrically consistent framework for representing coupled rotational and translational motion. Previous work has applied these methods to spacecraft estimation and control, including invariant filtering for attitude estimation~\cite{Gui2018}, relative pose and velocity estimation~\cite{mathavaraj2024consensus}, and $\mathrm{SE}(3)$-based spacecraft control~\cite{lee2016robust,lee2013decentralized}. For the problem considered here, attitude, velocity, and position form a natural error chain: attitude error perturbs the applied acceleration,
acceleration error produces velocity error, and velocity error produces position error. While $\mathrm{SE}(3)$ combines attitude and position, translational velocity remains outside the group. We therefore use $\mathrm{SE}_2(3)$, which places all three states within a single group element~\cite{Barrau2015} and allows this attitude--velocity--position coupling to be retained directly in the tracking-error dynamics. Related $\mathrm{SE}_2(3)$ log-linear formulations have also been applied to multirotor dynamics~\cite{lin2023multirotor}. Here, we consider continuously thrusting spacecraft dynamics with position-dependent gravity. The resulting error dynamics contain nonlinear contributions from both input mismatch and gravity mismatch; we retain the former explicitly and focus on isolating and bounding the latter.

In this work, we formulate the dynamics of a thrusting spacecraft on $\mathrm{SE}_2(3)$ and derive the corresponding logarithmic dynamics of a left-invariant tracking error. The formulation preserves the attitude--translation coupling caused by body-frame thrust and makes its effect on the translational tracking error explicit. We then isolate the contribution of position-dependent gravity and derive an analytical bound under the two-body gravity model. Together, these results provide a structured error model in which the coupled thrusting dynamics are retained while the effect of gravity is explicitly bounded. 

The primary contributions of this work are therefore twofold. First, we formulate the thrusting spacecraft dynamics in a mixed-invariant $\mathrm{SE}_2(3)$ form and derive the corresponding logarithmic tracking-error dynamics. Second, we derive an analytical bound on the position-dependent gravity contribution to these dynamics. The analytical results are evaluated numerically by comparing the logarithmic error dynamics with direct nonlinear propagation and verifying the gravity mismatch bound along a highly elliptical orbit.

The remainder of this paper is organized as follows. Section~\ref{sec:background} introduces the Lie-group notation required for the development. Section~\ref{sec:error_derivation} formulates the spacecraft tracking problem and derives the left-invariant tracking-error dynamics on $\mathrm{SE}_2(3)$. Section~\ref{sec:main_results} presents the logarithmic tracking-error dynamics and the two-body gravity bound. Section~\ref{sec:sim_results} evaluates these results numerically, and Section~\ref{sec:conclusion} summarizes the conclusions.


\section{Background}\label{sec:background}

This section summarizes the matrix Lie-group notation required for
subsequent development~\cite{Barfoot2017,Sola2018,BulloLewis2005}.
The matrix Lie group $\mathrm{SE}_2(3)\subset\mathbb{R}^{5\times 5}$, also referred to as the group
of double direct spatial isometries~\cite{Barrau2015}, is defined by
\begin{equation}
X =
\begin{bsmallmatrix}
R & v & p \\
\bm{0}_{1\times3} & 1 & 0 \\
\bm{0}_{1\times3} & 0 & 1
\end{bsmallmatrix}
\in\mathrm{SE}_2(3),
\quad
R\in\mathrm{SO}(3), \ v,p\in\mathbb{R}^3,
\label{eq:SE23_state}
\end{equation}
where $\mathrm{SO}(3)=\{R\in\mathbb{R}^{3\times3}:
R^\top R=I_3,\ \det R=1\}$ denotes the special orthogonal group.
Its Lie algebra $\mathfrak{se}_2(3)$ is the tangent space at the
identity, with vector coordinates ordered as
\begin{equation}
\xi =
\begin{bsmallmatrix}
\xi_p^\top & \xi_v^\top & \xi_R^\top
\end{bsmallmatrix}^{\top},
\qquad
\xi^\wedge =
\begin{bsmallmatrix}
\xi_R^\wedge & \xi_v & \xi_p\\
\bm{0}_{1\times3} & 0 & 0\\
\bm{0}_{1\times3} & 0 & 0
\end{bsmallmatrix},
\label{eq:algebra_coordinates}
\end{equation}
where $\xi_R,\xi_v,\xi_p\in\mathbb{R}^3$ are the rotational,
velocity, and position components, respectively. The wedge operator
maps vector coordinates to the Lie algebra, and the vee operator
provides the inverse mapping. The exponential and logarithmic maps
satisfy
\begin{equation}
X=\operatorname{Exp}(\xi^\wedge),
\qquad
\xi=\operatorname{Log}(X)^\vee,
\label{eq:exp_log}
\end{equation}
for $X\in\mathrm{SE}_2(3)$ within the domain of the logarithmic map.

For $A\in\mathfrak{se}_2(3)$, the group adjoint is
\begin{equation}
\operatorname{Ad}_X(A)=XAX^{-1}.
\label{eq:group_ad}
\end{equation}
For $\zeta\in\mathbb{R}^9$, the action of the group adjoint on vector
coordinates is denoted
$\operatorname{Ad}_X^\vee\in\mathbb{R}^{9\times9}$ and satisfies
\begin{equation}
\left(\operatorname{Ad}_X(\zeta^\wedge)\right)^\vee
=
\operatorname{Ad}_X^\vee\zeta.
\label{eq:group_ad_vector}
\end{equation}
The Lie-algebra adjoint is defined by the matrix commutator
\begin{equation}
\operatorname{ad}_A(B)
=
[A,B]
=
AB-BA,
\qquad
A,B\in\mathfrak{se}_2(3).
\label{eq:little_adjoint}
\end{equation}
For vector coordinates $\xi,\zeta\in\mathbb{R}^9$, its matrix
representation $\operatorname{ad}_\xi\in\mathbb{R}^{9\times9}$ is
defined by
\begin{equation}
\left(
\operatorname{ad}_{\xi^\wedge}(\zeta^\wedge)
\right)^\vee
=
\operatorname{ad}_\xi\zeta.
\label{eq:ad_vector_form}
\end{equation}

Finally, the left and right Jacobians relate the logarithmic-coordinate
rate $\dot\xi$ to tangent vectors on the group. For
$X=\operatorname{Exp}(\xi^\wedge)$,
\begin{equation}
\left(\dot X X^{-1}\right)^\vee
=
J_\ell(\xi)\dot\xi,
\qquad
\left(X^{-1}\dot X\right)^\vee
=
J_r(\xi)\dot\xi,
\label{eq:left_right_jacobian}
\end{equation}
where $J_\ell(\xi),J_r(\xi)\in\mathbb{R}^{9\times9}$ denote the
$\mathrm{SE}_2(3)$ left and right Jacobians, respectively.

\section{Spacecraft Dynamics and Error Formulation}\label{sec:error_derivation}
This section develops the spacecraft tracking formulation. 
We first define the actual and reference spacecraft dynamics and their left-invariant tracking error. We then express the spacecraft dynamics in a mixed-invariant form and derive the corresponding tracking-error dynamics on $\mathrm{SE}_2(3)$. Finally, we relate these dynamics to the logarithmic error coordinates.

\subsection{Spacecraft Tracking Problem}

Consider an actual spacecraft with inertial-frame position
$p\in\mathbb{R}^3$, inertial-frame velocity $v\in\mathbb{R}^3$,
and attitude $R\in\mathrm{SO}(3)$, where $R$ maps vectors from the
spacecraft body frame to the inertial frame. The equations of motion
are
\begin{equation}
    \dot p = v,
    \qquad
    \dot v = Ra + g(p),
    \qquad
    \dot R = R\omega^\wedge,
    \label{eq:spacecraft_dynamics}
\end{equation}
where $a\in\mathbb{R}^3$ is the commanded acceleration expressed in
the body frame, $\omega\in\mathbb{R}^3$ is the body-frame angular
velocity, and $g(p)\in\mathbb{R}^3$ is the gravitational acceleration
expressed in the inertial frame. Let $(\bar p,\bar v,\bar R)$ denote a reference spacecraft trajectory
satisfying \eqref{eq:spacecraft_dynamics},
where $\bar a$ and $\bar\omega$ are the corresponding reference
body-frame inputs. 

The term $Ra$ in \eqref{eq:spacecraft_dynamics} contains the
attitude--translation coupling. 
An attitude-tracking error can therefore
change the direction of the applied acceleration and produce velocity
and position tracking errors.
Let the actual and reference state of the spacecraft be defined as $X$ and $\bar{X}$, respectively. We define the left-invariant tracking error as
\begin{equation}
\eta
=
\bar X^{-1}X
=
\begin{bsmallmatrix}
\bar R^\top R &
\bar R^\top(v-\bar v) &
\bar R^\top(p-\bar p)\\
\bm{0}_{1\times3} & 1 & 0\\
\bm{0}_{1\times3} & 0 & 1
\end{bsmallmatrix}
\in\mathrm{SE}_2(3).
\label{eq:left_invariant_error}
\end{equation}
The error $\eta$ compares the actual state $X$ with the reference
state $\bar X$ in coordinates associated with the reference
configuration. The rotation $\bar R^\top R$ gives the actual attitude
relative to the reference attitude. The difference in quantities $v-\bar v$ and
$p-\bar p$ are first formed in the common inertial frame and are then
rotated by $\bar R^\top$ into the reference body-frame coordinates.

\begin{remark}
The error $\eta$ is left invariant because a left multiplication of
$X$ and $\bar X$ by any $Y\in\mathrm{SE}_2(3)$ leaves $\eta$
unchanged:
\begin{equation}
(Y\bar X)^{-1}(YX)=\bar X^{-1}X=\eta.
\label{eq:left_invariance}
\end{equation}
\end{remark}


The translational components of the errror in Lie algebra $\xi = \log(\eta)^\vee$ require one additional distinction. Let $e_p := \bar R^\top(p-\bar p), e_v := \bar R^\top(v-\bar v)$ denote the position and velocity differences in the reference frame. These quantities are not, in general,
equal to the logarithmic coordinates $\xi_p$ and $\xi_v$. To make this relationship explicit, we also require the left Jacobian of $\mathrm{SO}(3)$. 
From the $\mathrm{SE}_2(3)$ exponential map, the translational
components of the group error satisfy
\begin{equation}
e_p
=
J_{\ell,\mathrm{SO}(3)}(\xi_R)\xi_p,
\qquad
e_v
=
J_{\ell,\mathrm{SO}(3)}(\xi_R)\xi_v,
\label{eq:translational_log_relation}
\end{equation}
where $J_{\ell,\mathrm{SO}(3)}(\phi)\in\mathbb{R}^{3\times3}$ is the left-Jacobian for $SO(3)$. Therefore, $e_p$ and $e_v$ are the physical position and velocity
differences expressed in the reference frame, while $\xi_p$ and
$\xi_v$ are their corresponding logarithmic coordinates in Lie algebra. When
$\xi_R=0$,
$J_{\ell,\mathrm{SO}(3)}(0)=I_3$, and the two representations
coincide. 

\subsection{Mixed-Invariant Representation} \label{sec:mixed_invariant}

We next rewrite the spacecraft equations in a form that preserves the
frames in which the physical quantities naturally act. Define
$
M =
\begin{bsmallmatrix}
\bm{0}_{3\times3} & g(p) & \bm{0}_{3\times1}\\
\bm{0}_{1\times3} & 0 & 0\\
\bm{0}_{1\times3} & 0 & 0
\end{bsmallmatrix},
\quad
N =
\begin{bsmallmatrix}
\omega^\wedge & a & \bm{0}_{3\times1}\\
\bm{0}_{1\times3} & 0 & 0\\
\bm{0}_{1\times3} & 0 & 0
\end{bsmallmatrix},
$
and
$C =
\begin{bsmallmatrix}
\bm{0}_{3\times3} & \bm{0}_{3\times1} & \bm{0}_{3\times1}\\
\bm{0}_{1\times3} & 0 & 1\\
\bm{0}_{1\times3} & 0 & 0
\end{bsmallmatrix}.$
Here, $M$ contains the inertial-frame gravitational acceleration,
$N$ contains the body-frame acceleration and angular
velocity, and $C$ accounts for the position--velocity kinematics. With these definitions, the spacecraft dynamics in
\eqref{eq:spacecraft_dynamics} can be written as
\begin{equation}
\dot X
=
(M-C)X+X(N+C).
\label{eq:mixed_invariant_dynamics}
\end{equation}
To see how each term relates to the original spacecraft dynamics,
direct multiplication gives
\begin{equation}
XN =
\begin{bsmallmatrix}
R\omega^\wedge & Ra & \bm{0}_{3\times1}\\
\bm{0}_{1\times3} & 0 & 0\\
\bm{0}_{1\times3} & 0 & 0
\end{bsmallmatrix},
\quad
MX =
\begin{bsmallmatrix}
\bm{0}_{3\times3} & g(p) & \bm{0}_{3\times1}\\
\bm{0}_{1\times3} & 0 & 0\\
\bm{0}_{1\times3} & 0 & 0
\end{bsmallmatrix},
\end{equation}
and
\begin{equation}
XC-CX =
\begin{bsmallmatrix}
\bm{0}_{3\times3} & \bm{0}_{3\times1} & v\\
\bm{0}_{1\times3} & 0 & 0\\
\bm{0}_{1\times3} & 0 & 0
\end{bsmallmatrix}.
\label{eq:mixed_terms}
\end{equation}
Thus, $XN$ contains the attitude dynamics and the body-frame
acceleration, $MX$ contains the inertial-frame gravitational
acceleration, and $XC-CX$ produces the position kinematics
$\dot p=v$. 

For the reference trajectory, define
$\bar M :=
\begin{bsmallmatrix}
\bm{0}_{3\times3} & g(\bar p) & \bm{0}_{3\times1}\\
\bm{0}_{1\times3} & 0 & 0\\
\bm{0}_{1\times3} & 0 & 0
\end{bsmallmatrix},
$
and
$
\bar N :=
\begin{bsmallmatrix}
\bar\omega^\wedge & \bar a & \bm{0}_{3\times1}\\
\bm{0}_{1\times3} & 0 & 0\\
\bm{0}_{1\times3} & 0 & 0
\end{bsmallmatrix}.
$
The reference dynamics therefore have the same mixed-invariant structure as~\eqref{eq:mixed_invariant_dynamics}, with $X$, $M$, and
$N$ replaced by $\bar X$, $\bar M$, and $\bar N$, respectively.
Define the difference between the actual and reference matrices as
\begin{equation}
\widetilde M:=M-\bar M,
\qquad
\widetilde N:=N-\bar N.
\label{eq:matrix_mismatch}
\end{equation}
Using the definitions of $M$ and $N$, these
matrix differences have the explicit forms
\begin{equation}
\widetilde M =
\begin{bsmallmatrix}
\bm{0}_{3\times3} & g(p)-g(\bar p) & \bm{0}_{3\times1}\\
\bm{0}_{1\times3} & 0 & 0\\
\bm{0}_{1\times3} & 0 & 0
\end{bsmallmatrix},
\
\widetilde N =
\begin{bsmallmatrix}
(\omega-\bar\omega)^\wedge & a-\bar a &
\bm{0}_{3\times1}\\
\bm{0}_{1\times3} & 0 & 0\\
\bm{0}_{1\times3} & 0 & 0
\end{bsmallmatrix}.
\label{eq:matrix_mismatch_explicit}
\end{equation}

Both $\widetilde M$ and $\widetilde N$ belong to
$\mathfrak{se}_2(3)$, since they have the matrix structure defined in
\eqref{eq:algebra_coordinates}. We therefore denote their vector
coordinates in the $(p,v,R)$ ordering by
\begin{equation}
\widetilde m
:=
\widetilde M^\vee
=
\begin{bsmallmatrix}
0\\
g(p)-g(\bar p)\\
0
\end{bsmallmatrix},
\qquad
\widetilde n
:=
\widetilde N^\vee
=
\begin{bsmallmatrix}
0\\
a-\bar a\\
\omega-\bar\omega
\end{bsmallmatrix}.
\label{eq:vector_mismatch}
\end{equation}
Thus, $\widetilde M$ and $\widetilde N$ are the matrix mismatch terms
that appear in the group dynamics, while $\widetilde m$ and
$\widetilde n$ are the same quantities expressed in the vector
coordinates. The vector $\widetilde m$ contains the position-dependent gravity
mismatch, while $\widetilde n$ contains the body-frame input mismatch.

\subsection{Left-Invariant Tracking-Error Dynamics}
\label{sec:group_error}

We first derive the dynamics of the left-invariant tracking error $\eta=\bar X^{-1}X$ defined in~\eqref{eq:left_invariant_error}, and then express these dynamics in the logarithmic coordinates $\xi$ used in the subsequent analysis.

Differentiating $\eta=\bar X^{-1}X$ gives
\begin{equation}
\dot\eta
=
-\bar X^{-1}\dot{\bar X}\eta
+
\bar X^{-1}\dot X,
\label{eq:eta_derivative}
\end{equation}
where the derivative of the inverse follows from
$\bar X^{-1}\bar X=I$:
\begin{equation}
\frac{d}{dt}\left(\bar X^{-1}\right)
=
-\bar X^{-1}\dot{\bar X}\bar X^{-1}.
\label{eq:inverse_derivative}
\end{equation}

Substituting the actual and reference dynamics, and using
$X=\bar X\eta$, gives
\begin{equation}
\dot\eta
=
\bar X^{-1}\widetilde M\bar X\,\eta
+
\eta N-\bar N\eta
+
\eta C-C\eta.
\label{eq:eta_dynamics}
\end{equation}
Thus, \eqref{eq:eta_dynamics} describes the evolution of the
left-invariant tracking error directly on $\mathrm{SE}_2(3)$.
We now use the right-trivialization to represent the tangent vector $\dot\eta\in T_\eta\mathrm{SE}_2(3)$ in the Lie algebra by right multiplication with $\eta^{-1}$. The right trivialization results in
$
\dot\eta\eta^{-1}
\in
T_I\mathrm{SE}_2(3)
=
\mathfrak{se}_2(3).
$
Multiplying \eqref{eq:eta_dynamics} by $\eta^{-1}$ on the right and using
$N=\bar N+\widetilde N$ gives the following:
\begin{multline}
\dot\eta\eta^{-1}
=
\operatorname{Ad}_{\bar X^{-1}}(\widetilde M)
+
\left(\operatorname{Ad}_{\eta}-I\right)\bar N
+\\
\operatorname{Ad}_{\eta}(\widetilde N)
+
\eta C\eta^{-1}-C.
\label{eq:right_trivialized_error} 
\end{multline}

The four terms in~\eqref{eq:right_trivialized_error} separate the
sources of tracking-error evolution. The first contains the gravity
mismatch between the actual and reference spacecraft. The second term
captures the effect of the reference body-frame inputs on an existing
tracking error, while the third term contains the actual-reference
body-frame input mismatch. The final term carries the
position--velocity kinematic coupling introduced by $C$.
Finally, we get
the right-trivialized
relationship using~\eqref{eq:left_right_jacobian} as
$
\left(\dot\eta\eta^{-1}\right)^\vee
=
J_\ell(\xi)\dot\xi.
\label{eq:eta_log_relation}
$

\section{Main Results}
\label{sec:main_results}

In this section, we present the two main results of this work. First, we derive the logarithmic tracking-error dynamics
and show explicitly how attitude and translational errors are coupled
during body-fixed thrust. We then isolate the contribution of
position-dependent two-body gravity and derive an analytical bound on
this term.

\subsection{Logarithmic Tracking-Error Dynamics}
\label{sec:log_error_dynamics}

Before deriving the logarithmic tracking-error dynamics, we establish
two identities required to simplify the right-trivialized tracking-error
dynamics in~\eqref{eq:right_trivialized_error}. The first states
standard Lie-group adjoint and Jacobian identities. 
The second accounts for the
position--velocity kinematics introduced by the matrix $C$.

\begin{lemma}[Adjoint and Jacobian Identities]
\label{lem:adjoint_jacobian}
For $\eta=\operatorname{Exp}(\xi^\wedge)\in\mathrm{SE}_2(3)$,
$\xi\in\mathbb{R}^9$, with nonsingular $J_\ell(\xi)$ and $J_r(\xi)$, then the following hold:
\begin{align}
\operatorname{Ad}_{\eta}^{\vee}
&=e^{\operatorname{ad}_{\xi}}, 
\label{eq:adjoint_exp}\\
J_\ell^{-1}(\xi)(e^{\operatorname{ad}_{\xi}}-I_9)
&=\operatorname{ad}_{\xi}, \quad J_\ell^{-1}(\xi)e^{\operatorname{ad}_{\xi}}=J_r^{-1}(\xi).
\label{eq:jacobian_identities}
\end{align}
\end{lemma}

\begin{proof}
For any $\zeta\in\mathbb{R}^9$, let
$\zeta^\wedge\in\mathfrak{se}_2(3)$. From the definition of the group
adjoint,
\begin{equation}
    \operatorname{Ad}_{\eta}(\zeta^\wedge)
    =
    \eta\zeta^\wedge\eta^{-1}.
\end{equation}
Using $\eta=\operatorname{Exp}(\xi^\wedge)$, the matrix conjugation
expansion gives
\begin{align}
\eta\zeta^\wedge\eta^{-1}
&=
\zeta^\wedge
+
[\xi^\wedge,\zeta^\wedge]
+
\frac{1}{2!}
[\xi^\wedge,[\xi^\wedge,\zeta^\wedge]]
+\cdots .
\label{eq:adjoint_conjugation}
\end{align}
Since $\xi^\wedge,\zeta^\wedge\in\mathfrak{se}_2(3)$ and the Lie
algebra is closed under the commutator, each term in
\eqref{eq:adjoint_conjugation} belongs to
$\mathfrak{se}_2(3)$. Applying the vee map and using the
vector-coordinate representation of the Lie-algebra adjoint gives
\begin{align}
\left(
\operatorname{Ad}_{\eta}(\zeta^\wedge)
\right)^\vee
&=
\left(
I_9
+\operatorname{ad}_{\xi}
+\frac{1}{2!}\operatorname{ad}_{\xi}^{2}
+\cdots
\right)\zeta = e^{\operatorname{ad}_{\xi}}\zeta .
\end{align}
Because this holds for arbitrary $\zeta\in\mathbb{R}^9$,
\eqref{eq:adjoint_exp} follows.

Using the results from ~\cite{Barfoot2017, Sola2018}, we know that the left and right Jacobians satisfy
\begin{equation}
    J_\ell(\xi)\operatorname{ad}_{\xi}
    =
    e^{\operatorname{ad}_{\xi}}-I_9,
    \qquad
    J_r(\xi)
    =
    e^{-\operatorname{ad}_{\xi}}J_\ell(\xi).
    \label{eq:jacobian_relations}
\end{equation}
Left multiplying the first relation by $J_\ell^{-1}(\xi)$ gives the
first identity in~\eqref{eq:jacobian_identities}. Rearranging the
second relation gives the second identity, which completes the proof.
\end{proof}

The next identity is specific to the constant matrix $C$ introduced in Section~\ref{sec:mixed_invariant}. Since $C\notin\mathfrak{se}_2(3)$, Lie-algebra
operations cannot be applied directly to $C$. We therefore begin
with the matrix commutator defined in
Section~\ref{sec:background} and first establish that the resulting
matrix has the required Lie-algebra structure.

\begin{lemma}[Position--Velocity Kinematic Identity]
\label{lem:C_identity}
Let $\eta=\operatorname{Exp}(\xi^\wedge)\in\mathrm{SE}_2(3)$, where
$
\xi=
\begin{bsmallmatrix}
\xi_p^\top & \xi_v^\top & \xi_R^\top
\end{bsmallmatrix}^{\top}
\in\mathbb{R}^9,
$
and $J_\ell(\xi)$ is nonsingular. Let
$C\in\mathbb{R}^{5\times5}$ be the constant matrix defined in Section
\ref{sec:mixed_invariant}. Then
\begin{equation}
    \eta C\eta^{-1}-C\in\mathfrak{se}_2(3),
    \label{eq:C_lie_membership}
\end{equation}
and
\begin{equation}
    J_\ell^{-1}(\xi)
    \left(\eta C\eta^{-1}-C\right)^\vee
    =
    A_C\xi,
    \label{eq:C_identity}
\end{equation}
where
\begin{equation}
A_C
=
\begin{bsmallmatrix}
\bm{0}_{3\times3} & I_3 & \bm{0}_{3\times3}\\
\bm{0}_{3\times3} & \bm{0}_{3\times3} & \bm{0}_{3\times3}\\
\bm{0}_{3\times3} & \bm{0}_{3\times3} & \bm{0}_{3\times3}
\end{bsmallmatrix}.
\label{eq:AC_def}
\end{equation}
\end{lemma}

\begin{proof}
Because $C\notin\mathfrak{se}_2(3)$, we begin with the
matrix commutator. Direct multiplication gives
\begin{equation}
[\xi^\wedge,C]
=
\xi^\wedge C-C\xi^\wedge
=
\begin{bsmallmatrix}
\bm{0}_{3\times3} & \bm{0}_{3\times1} & \xi_v\\
\bm{0}_{1\times3} & 0 & 0\\
\bm{0}_{1\times3} & 0 & 0
\end{bsmallmatrix}.
\label{eq:C_commutator}
\end{equation}
The resulting matrix has the $\mathfrak{se}_2(3)$ structure defined
in~\eqref{eq:algebra_coordinates} by construction. Therefore,
\begin{equation}
    [\xi^\wedge,C]\in\mathfrak{se}_2(3),
\end{equation}
and its vector coordinates are well defined. Applying the vee map
gives
\begin{equation}
[\xi^\wedge,C]^\vee
=
\begin{bsmallmatrix}
\xi_v\\
0\\
0
\end{bsmallmatrix}
=
A_C\xi.
\label{eq:C_commutator_vector}
\end{equation}

Using $\eta=\operatorname{Exp}(\xi^\wedge)$, the matrix conjugation
expansion gives
\begin{align}
\eta C\eta^{-1}-C
&=
[\xi^\wedge,C]
+
\frac{1}{2!}
[\xi^\wedge,[\xi^\wedge,C]]
\nonumber\\
&\quad+
\frac{1}{3!}
[\xi^\wedge,[\xi^\wedge,[\xi^\wedge,C]]]
+\cdots .
\label{eq:C_conjugation}
\end{align}
The first commutator in~\eqref{eq:C_conjugation} has already been
shown to belong to $\mathfrak{se}_2(3)$. Since
$\xi^\wedge\in\mathfrak{se}_2(3)$ and the Lie algebra is closed under
the commutator, every subsequent term also belongs to
$\mathfrak{se}_2(3)$. Hence,
\eqref{eq:C_lie_membership} holds and the vee map in
\eqref{eq:C_identity} is well defined.

Applying the vee map to~\eqref{eq:C_conjugation} and using
\eqref{eq:C_commutator_vector} gives
\begin{align}
\left(\eta C\eta^{-1}-C\right)^\vee
&=
\left(
I_9
+\frac{1}{2!}\operatorname{ad}_{\xi}
+\cdots
\right)A_C\xi \\
&=
J_\ell(\xi)A_C\xi.
\label{eq:C_jacobian_relation}
\end{align}
Left multiplying by $J_\ell^{-1}(\xi)$ gives
\eqref{eq:C_identity}, which completes the proof.
\end{proof}

With these identities established, the tracking-error dynamics in
\eqref{eq:right_trivialized_error} can now be expressed directly in
the logarithmic coordinates $\xi$.

\begin{theorem}[Coupled Logarithmic Tracking-Error Dynamics]
\label{thm:log_error}
Let the actual and reference states $X,\bar{X}\in\mathrm{SE}_2(3)$ satisfy the mixed-invariant dynamics in
\eqref{eq:mixed_invariant_dynamics}, with the reference quantities obtained by replacing $X$, $M$, $N$ with $\bar{X}$, $\bar{M}$, $\bar{N}$, respectively. Let the left-invariant tracking error and its logarithmic coordinates be
defined by
\begin{equation}
    \eta=\bar X^{-1}X
    =
    \operatorname{Exp}(\xi^\wedge),
    \qquad
    \xi=\operatorname{Log}(\eta)^\vee,
\end{equation}
where $\xi\in\mathbb{R}^9$ and $J_\ell(\xi)$ and $J_r(\xi)$ are
nonsingular. Then the logarithmic tracking error satisfies
\begin{equation}
    \dot\xi
    =
    A\xi
    +
    J_r^{-1}(\xi)\widetilde n
    +
    d_g,
    \label{eq:error_dynamics}
\end{equation}
where
$
A
=
-\operatorname{ad}_{\bar n}+A_C
=
\begin{bsmallmatrix}
-\bar\omega^\wedge & I_3 & \bm{0}_{3\times3}\\
\bm{0}_{3\times3} & -\bar\omega^\wedge & -\bar a^\wedge\\
\bm{0}_{3\times3} & \bm{0}_{3\times3} & -\bar\omega^\wedge
\end{bsmallmatrix},
$
and
$ 
    d_g
    :=
    J_\ell^{-1}(\xi)
    \operatorname{Ad}_{\bar X^{-1}}^\vee
    \widetilde m.
    \label{eq:gravity_term}
$
\end{theorem}

\begin{proof}
Using~\eqref{eq:right_trivialized_error}, we can write
\begin{align}
J_\ell(\xi)\dot\xi
&=
\operatorname{Ad}_{\bar X^{-1}}^\vee\widetilde m
+
\left(e^{\operatorname{ad}_{\xi}}-I_9\right)\bar n
\nonumber\\
&\quad+
e^{\operatorname{ad}_{\xi}}\widetilde n
+
\left(\eta C\eta^{-1}-C\right)^\vee.
\label{eq:log_error_pre}
\end{align}
Multiplying by $J_\ell^{-1}(\xi)$ allows each of the final three
terms to be simplified using Lemmas~\ref{lem:adjoint_jacobian} and
\ref{lem:C_identity}. For the reference-input term,
\begin{align}
J_\ell^{-1}(\xi)
\left(e^{\operatorname{ad}_{\xi}}-I_9\right)\bar n
&=
\operatorname{ad}_{\xi}\bar n =
-\operatorname{ad}_{\bar n}\xi.
\label{eq:reference_input_simplification}
\end{align}
For the input-mismatch term,
\begin{equation}
J_\ell^{-1}(\xi)
e^{\operatorname{ad}_{\xi}}\widetilde n
=
J_r^{-1}(\xi)\widetilde n.
\label{eq:input_mismatch_simplification}
\end{equation}
Lemma~\ref{lem:C_identity} gives
\begin{equation}
J_\ell^{-1}(\xi)
\left(\eta C\eta^{-1}-C\right)^\vee
=
A_C\xi.
\label{eq:C_term_simplification}
\end{equation}
Finally, define the gravity contribution as
\begin{equation}
d_g
=
J_\ell^{-1}(\xi)
\operatorname{Ad}_{\bar X^{-1}}^\vee
\widetilde m.
\end{equation}
Collecting these terms gives
\begin{equation}
\dot\xi
=
\left(-\operatorname{ad}_{\bar n}+A_C\right)\xi
+
J_r^{-1}(\xi)\widetilde n
+
d_g,
\end{equation}
which gives~\eqref{eq:error_dynamics}.
\end{proof}

\begin{remark}
The matrix $A$ makes the nominal attitude--translation coupling
explicit through the term $-\bar a^\wedge\xi_R$ in the velocity-error
dynamics. During body-fixed thrust, attitude tracking error therefore
perturbs the applied acceleration direction, producing velocity error
that subsequently enters the position-error dynamics through $A_C$.
\end{remark}

\begin{remark}
If the body-frame input and gravity mismatches vanish,
$\tilde n=d_g=0$, then~\eqref{eq:error_dynamics} reduces to
$\dot{\xi}=A\xi$. For prescribed reference inputs $\bar a(t)$ and
$\bar\omega(t)$, the resulting logarithmic tracking-error dynamics
are therefore linear and generally time-varying.
\end{remark}

\subsection{Two-Body Gravity-Mismatch Bound}
\label{sec:gravity_bound}

Theorem~\ref{thm:log_error} separates the position-dependent gravity
contribution from the remaining logarithmic tracking-error dynamics as
\begin{equation}
d_g
=
J_\ell^{-1}(\xi)
\operatorname{Ad}_{\bar X^{-1}}^\vee\widetilde m.
\label{eq:dg_repeated}
\end{equation}
We now derive an upper bound on this term for the two-body gravity
model
\begin{equation}
g(p)
=
-\frac{\mu p}{\|p\|_2^3},
\label{eq:two_body_gravity}
\end{equation}
where $\mu>0$ is the gravitational parameter of the central body.

Recall that $\widetilde m$ contains only the difference between the
actual and reference gravitational accelerations. Evaluating the
adjoint action in~\eqref{eq:dg_repeated} gives
\begin{equation}
\operatorname{Ad}_{\bar X^{-1}}^\vee\widetilde m
=
\begin{bmatrix}
0\\
\bar R^\top\!\left[g(p)-g(\bar p)\right]\\
0
\end{bmatrix}.
\label{eq:gravity_adjoint}
\end{equation}
Because the vector in~\eqref{eq:gravity_adjoint} has only a velocity
component, repeated application of $\operatorname{ad}_{\xi}$
preserves this structure and acts on the velocity component through
$\xi_R^\wedge$. Consequently, the action of the full
$\mathrm{SE}_2(3)$ inverse left Jacobian $J_\ell^{-1}(\xi)$ on this
vector reduces to the action of the $\mathrm{SO}(3)$ inverse left
Jacobian $J_{\ell,\mathrm{SO}(3)}^{-1}(\xi_R)$ associated with~\eqref{eq:translational_log_relation}. The gravity contribution can therefore
be written as
\begin{equation}
d_g
=
\begin{bmatrix}
0\\
J_{\ell,\mathrm{SO}(3)}^{-1}(\xi_R)
\bar R^\top\!\left[g(p)-g(\bar p)\right]\\
0
\end{bmatrix}.
\label{eq:gravity_component}
\end{equation}
The magnitude of $d_g$ therefore depends on two quantities: the
difference between the gravitational accelerations at $p$ and
$\bar p$, and the inverse $\mathrm{SO}(3)$ left Jacobian associated
with the attitude-error coordinate $\xi_R$. Theorem~\ref{thm:gravity_bound}
bounds these quantities in terms of the logarithmic tracking error and
the reference orbital radius.


\begin{theorem}[Two-Body Gravity-Mismatch Bound]
\label{thm:gravity_bound}
Consider the actual and reference spacecraft under the two-body
gravity model~\eqref{eq:two_body_gravity}. Let
$r:=\|\bar p\|_2>0$ and $\theta:=\|\xi_R\|_2$, and assume $\|\xi_p\|_2<r$ and $0\leq\theta<\pi$, so that $\xi_R$ lies within the domain of the principal logarithm of
$\mathrm{SO}(3)$. Then the gravity contribution
$d_g$ in~\eqref{eq:gravity_component} satisfies
\begin{equation}
\|d_g\|_2
\le
\gamma(\theta)
\frac{
\mu\|\xi_p\|_2
\left(2r-\|\xi_p\|_2\right)
}{
r^2\left(r-\|\xi_p\|_2\right)^2
},
\label{eq:gravity_bound}
\end{equation}
where
\begin{equation}
\gamma(\theta)
:=
\begin{cases}
\dfrac{\theta/2}{\sin(\theta/2)}, & 0<\theta<\pi,\\
1, & \theta=0.
\end{cases}
\label{eq:gamma_theta}
\end{equation}
\end{theorem}

\begin{proof}
From~\eqref{eq:gravity_component}, submultiplicativity of the induced
matrix $2$-norm and $\|\bar R^\top\|_2=1$ give
\begin{equation}
\|d_g\|_2
\le
\left\|J_{\ell,\mathrm{SO}(3)}^{-1}(\xi_R)\right\|_2
\|g(p)-g(\bar p)\|_2.
\label{eq:gravity_norm_split}
\end{equation}
From the closed-form inverse $\mathrm{SO}(3)$ left Jacobian in
\cite[Eq.~(146)]{Sola2018}, its induced $2$-norm for
$\theta=\|\xi_R\|_2<\pi$ is
\begin{equation}
\left\|J_{\ell,\mathrm{SO}(3)}^{-1}(\xi_R)\right\|_2
=
\gamma(\theta).
\label{eq:inverse_jacobian_norm}
\end{equation}

Let $Dg(q) = \frac{\partial g(q)}{\partial q}$ denote the Jacobian of the two-body gravitational
acceleration with respect to $q$. Differentiating
\eqref{eq:two_body_gravity} gives
\begin{equation}
Dg(q)
=
-\frac{\mu}{\|q\|_2^3}
\left(
I_3-3\frac{qq^\top}{\|q\|_2^2}
\right).
\label{eq:gravity_gradient}
\end{equation}
With $\hat q=q/\|q\|_2$, the symmetric matrix
$I_3-3\hat q\hat q^\top$ has eigenvalue $-2$ in the radial direction
$\hat q$ and eigenvalue $1$ in each transverse direction. Its maximum
singular value is therefore $2$, so
\begin{equation}
\|Dg(q)\|_2=\frac{2\mu}{\|q\|_2^3}.
\label{eq:gravity_gradient_norm}
\end{equation}

Along the straight-line segment
$q(s)=\bar p+s(p-\bar p)$, $s\in[0,1]$, the fundamental theorem of
calculus and~\eqref{eq:gravity_gradient_norm} give
\begin{equation}
\|g(p)-g(\bar p)\|_2
\le
2\mu\|p-\bar p\|_2
\int_0^1\frac{ds}{\|q(s)\|_2^3}.
\label{eq:gravity_integral_bound}
\end{equation}
From the previously established relation between the physical and
logarithmic position errors,
\begin{equation}
\|p-\bar p\|_2\leq\|\xi_p\|_2,
\label{eq:physical_log_position_bound}
\end{equation}
since $\bar R$ is orthogonal and
$\|J_{\ell,\mathrm{SO}(3)}(\xi_R)\|_2=1$ for
$0\leq\theta<\pi$ (see~\cite{Sola2018}). The reverse triangle inequality
therefore gives
\[
\|q(s)\|_2
\geq r-s\|p-\bar p\|_2
\geq r-s\|\xi_p\|_2>0,
\]
where the final inequality follows from $\|\xi_p\|_2<r$. Hence,
\begin{align}
\|g(p)-g(\bar p)\|_2
&\le
2\mu\|\xi_p\|_2
\int_0^1
\frac{ds}{(r-s\|\xi_p\|_2)^3}
\nonumber\\
&=
\frac{
\mu\|\xi_p\|_2
\left(2r-\|\xi_p\|_2\right)
}{
r^2\left(r-\|\xi_p\|_2\right)^2
}.
\label{eq:log_gravity_bound}
\end{align}
Combining this result with
\eqref{eq:gravity_norm_split}--\eqref{eq:inverse_jacobian_norm}
establishes~\eqref{eq:gravity_bound}.
\end{proof}

\begin{remark}
Theorem~\ref{thm:gravity_bound} provides a pointwise bound on the
nonlinear gravity contribution in terms of the logarithmic position
error $\|\xi_p\|_2$, the attitude-dependent factor $\gamma(\theta)$,
and the reference orbital radius $r=\|\bar p\|_2$. Thus, the bound may
vary along the reference trajectory, with its dependence on $r$
accounting for the stronger gravity gradient at smaller orbital radii.
A uniform bound over a prescribed trajectory may instead be obtained
by bounding these quantities over the trajectory.
\end{remark}

\begin{remark}
The tracking-error formulation is not restricted to two-body gravity;
more general gravitational models, including additional bodies, can be
incorporated through the mismatch term $d_g$. Theorem~\ref{thm:gravity_bound},
however, provides a closed-form bound specifically for the two-body
model in~\eqref{eq:two_body_gravity}; extension to an $n$-body model
requires bounding the additional gravitational acceleration differences.
\end{remark}


\section{Numerical Simulation}
\label{sec:sim_results}
\begin{figure}[t]
\centering
\includegraphics[width=0.36\textwidth]{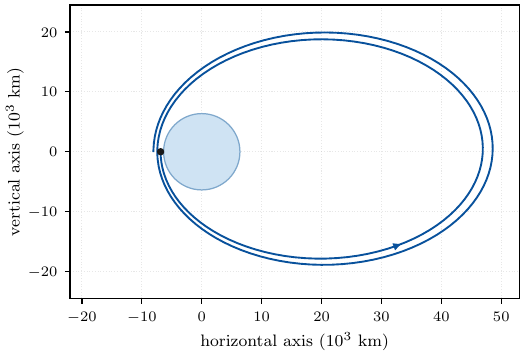}
\caption{Illustration of the Molniya-like orbit used for validation of the proposed framework.}
\label{fig:scenario}
\end{figure}

\begin{figure*}[t]
\centering
\includegraphics[width=0.9\textwidth]{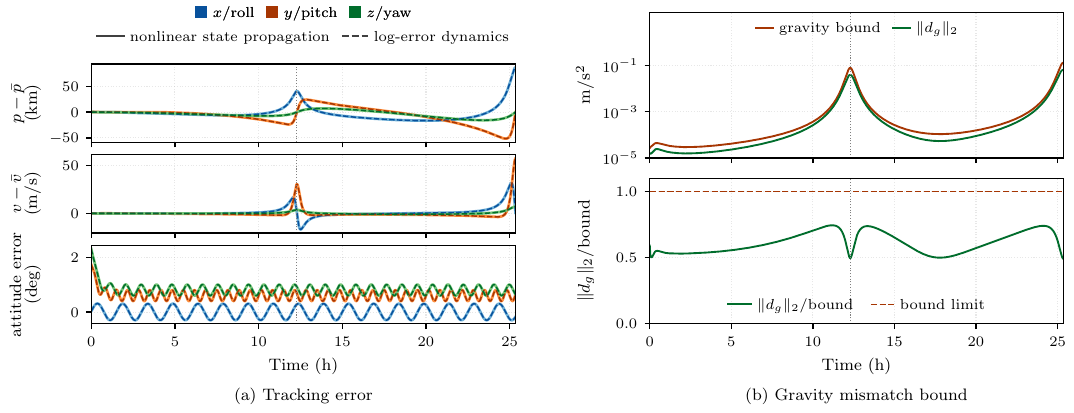}
\caption{Numerical validation over a two-revolution low-thrust
maneuver. \textbf{Left:} Position, velocity, and attitude tracking errors obtained from
nonlinear spacecraft propagation (solid) and the logarithmic error
dynamics of Theorem~\ref{thm:log_error} (dashed).
\textbf{Right:}  Gravity contribution $\|d_g\|_2$ and the analytical bound from
Theorem~\ref{thm:gravity_bound} (top), and their ratio with the
bound threshold of unity (bottom).}
\label{fig:results}
\end{figure*}
The numerical study evaluates the logarithmic tracking-error dynamics
in Theorem~\ref{thm:log_error} and the two-body gravity bound in
Theorem~\ref{thm:gravity_bound}. The reference spacecraft begins
at perigee of a Molniya-like orbit with semi-major axis
$26{,}521~\mathrm{km}$, eccentricity $0.741$, perigee altitude
$500~\mathrm{km}$, apogee altitude $39{,}800~\mathrm{km}$, and period
$11.94~\mathrm{h}$, as shown in Fig. \ref{fig:scenario}. The corresponding Keplerian state at perigee
defines the initial reference position and velocity
$\bar p(0)$ and $\bar v(0)$, which are
subsequently propagated using~\eqref{eq:spacecraft_dynamics}. A
constant body-frame acceleration of $2\times10^{-3}~\mathrm{m/s^2}$
is applied along the body $x$-axis for two orbital revolutions. The
reference attitude maintains the body $x$-axis along the velocity
direction and the body $z$-axis along the orbit normal.

The actual and reference spacecraft receive the same commanded
body-frame acceleration, $a=\bar a$. The actual spacecraft is
initialized with position and velocity offsets
$[10,\ 4,\ -2]^\top~\mathrm{m}$ and
$[1.0,\ -0.4,\ 0.2]^\top~\mathrm{cm/s}$, respectively. A prescribed
rotational logarithmic error consists of an initial $0.05~\mathrm{rad}$
error that decays with a $20$-min time constant toward a residual
$1^\circ$ thrust-vector misalignment, together with
$0.2^\circ-0.3^\circ$ oscillations with periods of
$40-90~\mathrm{min}$. The attitude-tracking error is prescribed directly rather than
generated by a specific attitude-control model. Consistent
with the left-invariant error definition, the actual attitude is
constructed as
$R=\bar R\operatorname{Exp}(\xi_R^\wedge)$, with
$\omega=(R^\top\dot R)^\vee$. Thus, although $a=\bar a$, the
time-varying attitude error continuously perturbs the inertial thrust
direction. The initial logarithmic error is computed from
$\xi(0)=\operatorname{Log}(\bar X(0)^{-1}X(0))^\vee$.

The nonlinear spacecraft states and logarithmic error dynamics are
integrated using an adaptive Dormand--Prince Runge--Kutta 4(5) integrator
with relative tolerance $10^{-13}$ and sampled every $30~\mathrm{s}$.
For the first validation, the actual and reference states are
independently propagated using~\eqref{eq:spacecraft_dynamics}.
Separately, $\xi$ is propagated using~\eqref{eq:error_dynamics} with
the same reference trajectory, inputs, and initial condition and
mapped back to the physical states through
$X=\bar X\operatorname{Exp}(\xi^\wedge)$.

Figure~\ref{fig:results}(a) shows that the two calculations remain
visually indistinguishable as the position error grows from
approximately $11~\mathrm{m}$ to $86~\mathrm{km}$. The maximum
discrepancies are $0.26~\mathrm{mm}$ in position,
$1.7\times10^{-7}~\mathrm{m/s}$ in velocity, and
$1.9\times10^{-13}~\mathrm{rad}$ in attitude, numerically validating
the logarithmic tracking-error dynamics for the simulated maneuver.

Finally, $\|d_g(t)\|_2$ is evaluated against the analytical bound in
Theorem~\ref{thm:gravity_bound}. The required conditions
$\|\xi_R\|_2<\pi$ and $\|\xi_p\|_2<\|\bar p\|_2$ are satisfied
throughout the maneuver. Figure~\ref{fig:results}(b) compares the
gravity contribution directly with the analytical bound and shows
their ratio. The ratio remains below unity throughout the maneuver,
ranging from approximately $0.49$ to $0.75$. Thus, the analytical bound contains the gravity contribution throughout the simulated highly elliptical trajectory, consistent
with Theorem~\ref{thm:gravity_bound}.

\section{Conclusion and Future Work}\label{sec:conclusion}
In this work, we presented the coupled logarithmic tracking-error dynamics for a thrusting spacecraft on $SE_2(3)$, making the effect of attitude error on translational motion explicit and separating the structured linear dynamics from input and gravity mismatch terms. Subsequently, we isolated the position-dependent two-body gravity contribution and obtained an analytical bound. We validated the numerical propagation of the derived coupled error dynamics against the nonlinear spacecraft tracking error dynamics and confirmed the gravity bound throughout a highly elliptical thrusting trajectory. In the future, we will explore how our modeling approach can be utilized to provide a structured error model for reachability and uncertainty analysis of attitude-coupled spacecraft maneuvers.

\section*{Acknowledgements}
This material is based on research sponsored by DARPA under agreement number FA8750-24-2-0500. The U.S. Government is authorized to reproduce and distribute reprints for Governmental purposes notwithstanding any copyright notation thereon. The authors acknowledge the use of Generative AI to develop code and help with grammar checks. However, the authors edited and verified all technical content, derivations, and results and take full responsibility.

\bibliographystyle{IEEEtran}
\bibliography{references}

\end{document}